\documentclass[11pt]{article}%

\usepackage{hyperref}
\usepackage{amssymb}
\usepackage{amsfonts}
\usepackage{geometry}
\usepackage{amsmath}
\usepackage{graphicx}%
\usepackage{color}
\usepackage{subcaption}
\providecommand{\U}[1]{\protect\rule{.1in}{.1in}}
\newtheorem{theorem}{Theorem}
\newtheorem{h}{Hypothesis}

\newtheorem{example}[theorem]{Example}

\newtheorem{remark}[theorem]{Remark}

\newenvironment{proof}[1][Proof]{\noindent\textbf{#1.} }{\ \rule{0.5em}{0.5em}}
\begin{document}

\title{On the skew and curvature of the implied and local volatilities}
\author{Elisa Al\`{o}s\thanks{Universidad Pompeu Fabra and Barcelona School of Economics, Barcelona, Spain. }
\and David Garc\'ia-Lorite\thanks{CaixaBank, Madrid, Spain}
\and Makar Pravosud\thanks{Universidad Pompeu
Fabra, Barcelona, Spain}}
\maketitle

\begin{abstract}
In this paper, we study the relationship between the short-end of the local and the implied volatility surfaces. Our results, based on Malliavin calculus techniques, recover the recent $\frac{1}{H+3/2}$ rule (where $H$ denotes the Hurst parameter of the volatility process) for rough volatilitites (see Bourgey, De Marco, Friz, and Pigato (2022)), that states that the short-time skew slope of the at-the-money implied volatility is $\frac{1}{H+3/2}$ of the corresponding slope for local volatilities. Moreover, we see that the at-the-money short-end curvature of the implied volatility can be written in terms of the short-end skew and curvature of the local volatility and vice versa. Additionally, this relationship depends on $H$.
\bigskip


\bigskip


\end{abstract}

\section{Introduction}
Local volatilities  are the main tool in real market practice (see Dupire (1994)), since they are the simplest models that capture the empirical implied volatility surface. They are an example of mimicking process (see Gy\"ongy (1986)), in the sense that they are one-dimensional models that can reproduce the marginal distributions of asset prices $S_t$.
In a local volatility model the volatility process is a deterministic function $\sigma(t,S_t)$ of time and the underlying asset price. The values of this function can be computed via Dupire's formula (see again Dupire (1994)).The plot of this function $\sigma$, is called the local volatility surface.

One challenging problem in this context is the study of the relationship between the implied and local volatilities. Even when both surfaces are similar, we can easily notice that the short-end local volatility smiles are more pronounced than the implied volatility smiles. In fact, some empirical studies (see Derman, Kani, and Zou (1996))  show that for short and intermediate maturities the at-the-money implied volatility skew is approximately half of the skew of the local volatility. This property is widely known as the {\it one-half rule}. 

There have been many attempts to address this phenomena from the analytical point of view. Classical proofs of this property for stochastic volatility models can be found in the literature. For example, in Derman, Kani, and Zou (1996) or in Gatheral (2006), this property is deduced from the representation of the implied volatilities as averaged local volatilities. In Lee (2001), the expansions of the implied and local volatility allow to proof this property by a direct comparison. In Al\`os and Garc\'ia-Lorite (2021), Malliavin calculus techniques give a representation of the short-limit at-the-money (ATM) implied volatility skew as an averaged local volatility skew, from which the one-half rule follows directly.

Nevertheless, recent studies (see Bourgey, De Marco, Friz, and Pigato (2022)) state that the one-half rule does not hold for the rough volatility models, where the volatility process is driven by a fBm with Hurst parameter $H<\frac12$. More precisely, the ATM short-end implied volatility skew is $\frac{1}{H+\frac32}$ of the ATM short-end local volatility skew. This result can be obtained via large deviations techniques.

In this paper we show how Malliavin calculus techniques lead us to an easy proof of $\frac{1}{H+\frac32}$ rule. Additionally, we study the relationship between the curvature of implied and local volatilities. In particular, we show how the ATM short-end implied volatility curvature can be represented in terms of the ATM short-end limit skew and the curvature of the corresponding local volatility, and vice versa. Moreover, we explain why $\frac{1}{H+\frac32}$ rule is compatible with real market practices, where the short-end local volatility is usually computed via a regularization procedure.

This paper is organized as follows. In Section 2, we introduce the basic concepts and definitions of Malliavin calculus. Section 3 is devoted to the study of the at-the-money skew slope for a local volatility model. In Section 4 we analyze the relation between the curvatures of the rough and local volatilities. Finally, we present some numerical examples in Section 5.

\section{Basic concepts of Malliavin calculus}
In this section we recall the key tools of Malliavin calculus that we use in this paper. We refer to Al\`os and Garc\'ia-Lorite (2021) for a deeper introduction to this topic and its applications to finance. 

\subsection{Basic definitions}
If
\(Z=(Z_t)_{t \in [0,T]}\) is a standard Brownian motion, we denote by ${\cal S}$  the set of random variables of the form
\begin{equation}
\label{eq:2.1}
F=f(Z(h_{1}),\ldots ,Z(h_{n})),  
\end{equation}
where $h_{1},\ldots ,h_{n}\in L^2([0,T])$, $Z(h_i)$ denotes the Wiener integral of $h_i$, for $i=1,..,n$, and $f\in C_{b}^{\infty }(\mathbb{R}^n) $ 
(i.e., $f$ and all its partial derivatives are bounded). If $F\in \cal S$, the Malliavin 
derivative of $F$ with respect to $Z$, $D^ZF$,  is defined
as the stochastic process in  $L^{2}(\Omega \times [0,T])$ given by
\begin{equation*}
D_{s}^ZF=\sum_{j=1}^{n}{\frac{\partial f}{\partial x_{j}}}(W(h_{1}),\ldots ,W_(h_{n}))(s)h_j(s).
\end{equation*}
Moreover, for $m\ge 1$, we can define the iterated Malliavin derivative operator $D^{m,Z}$, as
\begin{equation*}
D_{s_{1},\ldots ,s_{m}}^{m,Z}F=D_{s_{1}}^Z\ldots D_{s_{m}}^ZF,\qquad s_{1},\ldots
,s_{m}\in [0,T].
\end{equation*}
The operators $D^{m,Z}$ are closable  in $
L^{2}(\Omega )$ and we denote by ${\mathbb{D}}_Z^{n,2}$ the
closure of ${\cal S}$ with respect to the norm
$$
||F||_{n,2}=\left( E\left| F\right|
^{p}+\sum_{i=1}^{n}E||D^{i,Z}F||_{L^{2}([0,T]^{i})}^{2}\right) ^{\frac{1}{2}}.
$$
Notice that the Malliavin derivative operator satisfies the {\it chain rule}. That is, given $f\in \mathcal{C}^{1,2}_Z$, and $F\in \mathbb{D}^{1,2}_Z$, the random variable $f(F)$ belongs to $\mathbb{D}^{1,2}_Z$, and $D^Zf(F)=f'(F)D^WF$.
We will also make use of the notation $\mathbb{L}^{n,2}=\mathbb{D}_{Z}^{n,p}(L^2([0,T])).$\\

The adjoint of the derivative operator $D^Z$ is the divergence 
operator $\delta^Z $, which coincides with the Skorohod integral. Its domain, denoted by Dom $\delta $, is the set of processes 
$u\in L^{2}(\Omega \times [0,T])$ such that there exists a random variable
$\delta^Z (u)\in L^{2}(\Omega )$ such that
\begin{equation}\label{eq:ibpf}
E(\delta^Z (u)F)=E\left(\int_{0}^{T}(D_{s}^ZF)u_{s}ds\right),\qquad \hbox{\rm for every}
\;\;F\in {\cal S}.  
\end{equation}
We use the notation $\delta^Z (u)=\int_{0}^{T}u_{s}dZ_{s}$.  It is well known that $\delta $ is an extension of the It\^o integral. That is, $\delta$ applied to an adapted and square integrable processes,  coincides with the classical It\^o  integral. Moreover, the space 
${\mathbb{L}}^{1,2}$ is included in the domain of $\delta$. \\

From the above relationship between the operators $D^Z$ and $\delta^Z$, it is easy to see that, for an It\^o process of the form
$$
X_t=X_0+\int_0^t a_sds+\int_0^t b_s dZ_s,
$$
where $a$ and $b$ are adapted processes in $\mathbb{L}^{1,2}_Z$, its Malliavin derivative is given by
\begin{equation}
\label{derito}
D_u^ZX_t=\int_0^t D^Z_ua_sds+b_u{\bf1}_{[0,t]}(u)+\int_0^t D^Z_u b_s dZ_s.
\end{equation}
Then, if we consider the equation of the form
$$
X_t=X_0+\int_0^t a (s,X_s)ds+\int_0^t b(s,X_s) dZ_s,
$$
where $a(s,\cdot)$ and $b(s,\cdot)$ are differentiable functions with  uniformly bounded derivatives,
a direct application of (\ref{derito}) allows us to see that
\begin{equation}
\label{dereq}
D_u^ZX_t=\int_u^t \frac{\partial a}{\partial x} (s,X_s)D_u^Z X_sds+b(u,X_u)+\int_u^t \frac{\partial b}{\partial x} (s,X_s) D_u^Z X_sdZ_s.
\end{equation}
Notice that the above equality also holds, for example, if $a$ and $b$ are globally Lipschitz functions with
polynomial growth (see Theorem 2.2.1 in Nualart (2006)), replacing $\frac{\partial a}{\partial x}$ and $\frac{\partial b}{\partial x}$ by an adequate processes.

\subsection{Malliavin calculus for local volatilities}\label{local-volatility-model}
Consider a local volatility model of the form
\begin{equation}
\label{localvol}
S_t=S_0+\int_0^t \sigma(u,S_u)S_u dW_u,
\end{equation}
where $W$ is a Brownian motion defined in a probability space $(\Omega, \mathcal{F}, P)$, and  where we take the interest rate $r=0$ for the sake of simplicity. In order to study the Malliavin differentiability of the asset price $S$, we introduce consider the following hypotheses

\begin{h} \label{Hyp1} $\sigma$ is continuous and there exist two constants $c_1,c_2>0$ such that, for all $(t,x)\in [0,T]\times \mathbb{R}$, $$c_1<\sigma (t,x)<c_2.$$
\end{h}
\begin{h} \label{Hyp2}
There exist two constants $C>0$ and $\gamma\in (0,\frac12)$ such that, for all $t\in [0,T]$, 
$\sigma (t,\cdot)$ is differentiable and $$S_t\partial_x \sigma (t,S_t)\leq Ct^{-\gamma}.$$
\end{h}
Notice that Hypothesis \ref{Hyp2} is equivalent to the following $$\partial_x \hat\sigma (t,S_t)\leq Ct^{-\gamma},$$ where $\hat\sigma$ denotes the local volatility in terms of the log-price (that is, $\hat\sigma(t,x)=\sigma(t,e^x)$). 
According to (\ref{dereq}), the Malliavin derivative of the random variable $S_t$ is given by

$$
D_r^WS_t=\sigma(r,S_r)S_r+\int_r^t a(u,S_u) D^W_rS_udW_u,
$$
where $r<t$ and $a(u,S_u):\partial_x\sigma(u,S_u)S_u+\sigma(u,S_u)$. Notice that as $\sigma$ is bounded and \ref{Hyp1} holds for all $p\geq 2$, the Burkh\"older's inequality gives us that
\begin{eqnarray}
E(D_r^WS_t)^p&\leq &c_p\left[ E(\sigma(r,S_r)S_r)^p+\left(\int_r^t u^{-2\gamma}E _r(D^W_rS_u)^2du\right)^\frac{p}{2}\right]\nonumber\\
&\leq & c_p \left[(E(S_r)^p+\left(\int_r^t u^{-2\gamma}\right)^{\frac{p}{p-2}}\left(\int_r^t u^{-2\gamma}E _r(D^W_rS_u)^pdu\right)\right],
\end{eqnarray}
for some positive constant $c_p$. This implies, by Gronwall's lemma, that for all $t<T$ $E_r(D_r^WS_t)^p\leq c_p(\sigma(r,S_r)S_r)^p$, which leads to $E_r(D_r^WS_t)^p\leq c_p'$, for some positive constant $c_p'$. Moreover, a direct application of It\^o's lemma gives us that
\begin{equation}
\label{derlv}
D_r^WS_t=\sigma(r,S_r)S_r\exp\left(-\frac12 \int_r^t a^2(u,S_u)du+\int_r^t a(u,S_u) dW_u\right)
\end{equation}

In order to study the second Malliavin derivative of $S$, we introduce the following hypothesis
\begin{h} \label{Hyp3}
There exist two constants $C>0$ and $\gamma\in(0,\frac14)$ such that, for all $t\in [0,T]$, 
$\sigma (t,\cdot)$ is twice differentiable and $$S_t^2\partial_{xx}^2 \sigma (t,S_t)\leq Ct^{-2\gamma}.$$
\end{h}
Then, under Hypotheses \ref{Hyp1},  \ref{Hyp2}, and \ref{Hyp3}, the second 
Malliavin derivative of the random variable $S_t$ is given by
\begin{eqnarray}
D_\theta^W D_r^WS_t&=&a(r,S_r)D_\theta^W S_r+\int_r^t \partial_x a(u,S_u) D^W_r S_u D^W_\theta S_udW_u\nonumber\\
&+&\int_r^t   a(u,S_u) D^W_r D^W_\theta S_udW_u
\end{eqnarray}
for $\theta<r$. Notice that as  Hypotheses \ref{Hyp1},  \ref{Hyp2}, and \ref{Hyp3} hold the stochastic integrals are well defined. Consequently, a similar arguments, as in the case of the first Malliavin derivative, give us that $E(D_\theta D_r^WS_t)^p\le c_p$ for some positive constant $c_p$. 
Moreover,
\begin{eqnarray}
&&D_\theta^WD_r^WS_t\\
&&=a(r,S_r)D_\theta^W S_r \exp\left(-\frac12 \int_r^t a^2(u,S_u)du+\int_r^t a(u,S_u) dW_u\right)\nonumber\\
&&+\sigma(r,S_r)S_r\exp\left(-\frac12 \int_r^t a^2(u,S_u)du+\int_r^t a(u,S_u) dW_u\right)\nonumber\\
&&\times\left(-\frac12 \int_r^t D_\theta^W(a^2(u,S_u))du+\int_r^t D_\theta^W (a(u,S_u)) dW_u\right)\\
&&=a(r,S_r)\sigma(\theta,S_\theta)S_\theta\exp\left(-\frac12 \int_\theta^ta^2(u,S_u)du+\int_\theta^t
a(u,S_u) dW_u\right)\nonumber\\
&&+\sigma(r,S_r)S_r\exp\left(-\frac12 \int_r^t a^2(u,S_u)du+\int_r^t a(u,S_u) dW_u\right)\nonumber\\
&&\times\left(-\frac12 \int_r^t D_\theta^W(a^2(u,S_u))du+\int_r^t D_\theta^W (a(u,S_u)) dW_u\right).
\end{eqnarray}
\begin{remark} Consider $t>0$. 
Notice that under the Hypotheses \ref{Hyp1},  \ref{Hyp2}, and \ref{Hyp3} for $\theta, r\to t$, 
\begin{equation}
\label{first}
D_r^WS_t\to \sigma(t,S_t)S_t,
\end{equation}
and
\begin{equation}
\label{second}
D_\theta^WD_r^WS_t \to a(t,S_t)\sigma(t,S_t)S_t=\partial_x\sigma(t,S_t)S_t^2+\sigma^2(t,S_t)S_t,
\end{equation}
a.s.
\end{remark}

Finally, let us consider the following hypothesis
\begin{h} \label{Hyp4}
There exist two constants $C>0$ and $\gamma\in(0,\frac13)$ such that, for all $t\in [0,T]$, 
$\sigma (t,\cdot)$ is twice differentiable and $$S_t^3\partial_{xxx}^3 \sigma (t,S_t)\leq Ct^{-3\gamma}.$$
\end{h}
Using the same arguments as before, one can see that under this hypothesis $E|D_u^WD_\theta D_r^WS_t|^p\le c_p$ for some positive constant $c_p$, and for all $u<\theta<r<t$. 

\section{The skew}\label{skew}

Consider a risk-neutral probability model for asset prices of the form
\begin{equation}
\label{themodel}
dS_t=\sigma_t S_t (\rho dW_+\sqrt{1-\rho^2} dB_t),
\end{equation}
where we assume the interest rate to be zero, $\rho\in [-1,1]$, $W$ and $B$ are two independent Brownian motions. We denote by  $\mathcal{F}^W$ and  $\mathcal{F}^B$ the 
$\sigma$-algebra generated by $W$ and $B$, respectively, and $\mathcal{F}:=\mathcal{F}^W \vee \mathcal{F}^B$. Moreover, $\sigma$ is a stochastic process adapted to the filtration generated by $W$. Notice that (\ref{themodel}) includes the cases of local volatilities, classical stochastic volatility models (where $\sigma$ is assumed to be a diffusion) and fractional volatilities (where $\sigma$ is driven by a fractional Brownian motion). In particular, it includes the case of rough volatilities (fractional volatilities with Hurst parameter $H<\frac12$).\\

We denote by $C_{BS}(t, T,X_t,k,\sigma)$ the Black-Scholes price at time $t$ of a European call  with time to maturity $T$, log-strike price $k:=\ln K$, and by  $I_t(T,k)$ the corresponding Black-Scholes implied volatility. Moreover, we denote as $k^*= \ln(S_t)$ the at-the-money strike.\\

Now we introduce the general results on the short-end behaviour of the volatility skew that we use in this section.

\subsection{Short-end limit of the skew slope}
Let us fix $t\in [0,T]$ and consider the following hypotheses.

\begin{h} \label{Hyp5}
The process $\sigma=(\sigma_s)_{s\in [t,T]}$ is positive and continuous a.s., and satisfies that for all $s \in [t,T]$,
$$
c_1 \leq \sigma_s\leq c_2,
$$
for some positive constants $c_1$ and $c_2$.
\end{h}
\begin{h}\label{Hyp6}
$\sigma\in \mathbb{L}^{2,2}_W$, and there exist $C>0$ and $H\in (0,1)$ such that for all $t\leq \tau \leq \theta \leq r \leq u \leq T$ and for all $p>0$
\begin{align*}
\begin{split}
(E |(D_{\theta}^{W}\sigma_r^2 )|^p)^\frac{1}{p}&\leq C_t(r-\theta)^{H-\frac12},\\
 (E|(D_{\theta}^{W}D_{r}^{W}\sigma_u^2 )|^p)^\frac{1}{p}&\leq C_t(u-r)^{H-\frac12}(u-\theta)^{H-\frac12},
 \end{split}
\end{align*}
for some positive constant $C_t$.
\end{h}

\begin{h}\label{Hyp7}
Hypothesis \ref{Hyp4} holds and the following quantity
\begin{align*}
\begin{split}
\frac{1}{(T-t)^{ \frac{3}{2}+ H}} \mathbb{E}_t\int_t^T \left( \int_s^T D_s^{W}\sigma_u^2du \right)ds 
\end{split}
\end{align*}
where $ \mathbb{E}_t$ denotes the expectation with respect to  $\mathcal{F}_t$, 
has a finite limit as $T \to t$
\end{h}

Under the above hypotheses, the at-the-money  short-end skew slope of the implied volatility can be computed from the following adaptation of Theorem 6.3 in Al\`os, Le\'on and Vives (2007) (see also Theorem 7.5.2 in Al\`os and Garc\'ia-Lorite (2021)).
\begin{theorem} 
\label{tskew}
Assume that Hypotheses \ref{Hyp5}, \ref{Hyp6} and \ref{Hyp7}, hold for some $t \in [0,T]$. Then 
\begin{equation}
\lim_{T \to t} (T-t)^{\frac{1}{2}-H}\partial_kI_t(T,k^{*}) = \frac{\rho}{ 2\sigma_t^2}\lim_{T \to t}\frac{1}{ (T-t)^{\frac{3}{2}+H}} \mathbb{E}_t\left( \int_t^T\left(\int_r^TD_r^{W}\sigma_u^2du \right) dr \right).
\label{skewIV}
\end{equation}
\label{recallth}
\end{theorem}

\subsection{The skew in local volatility models}

Now let us consider a local volatility model as in (\ref{localvol}). First of all we recall the case of differentiable local volatilities with bounded derivatives.

\subsubsection{Regular local volatilities}
If the local volatility function is bounded with bounded derivatives, the model (\ref{localvol}) satisfies Hypotheses \ref{Hyp5}, \ref{Hyp6} and \ref{Hyp7} for any $t\in [0,T]$, with $\sigma_u=\sigma(u,S_u)$ and $\gamma=0$. Then Theorem \ref{tskew} gives us that, for every fixed $t$, $\lim_{T \to t} \partial_kI_t(T,k^{*})$ is finite. Moreover, as 
$$
D_r^{W}\sigma^2(u,S_u)=2\sigma(u,S_u)\partial_x \sigma(u,S_u)D_r^WS_u
$$
and $D_r^WS_u\to \sigma(u,S_u)S_u$ as $r\to u$, we get  (see Al\`os and Garc\'ia-Lorite (2021))
\begin{eqnarray}
&&\lim_{T \to t} \partial_kI_t(T,k^{*}) \nonumber\\
&&= \frac{1}{ 2\sigma^2(t,S_t)}\lim_{T \to t}\frac{1}{ (T-t)^{2}} \mathbb{E}_t\left( \int_t^T\left(\int_r^TD_r^{W}\sigma^2(u,S_u)du \right) dr \right)\nonumber\\
&&= \frac{1}{ \sigma^2(t,S_t)}\lim_{T \to t}\frac{1}{ (T-t)^{2}} \mathbb{E}_t\left( \int_t^T\left(\int_r^T\partial_x \sigma(u,S_u) \sigma (u,S_u)D_r S_udu \right) dr \right)\nonumber\\
&&= \lim_{T \to t}\frac{1}{ (T-t)^{2}} \mathbb{E}_t\left( \int_t^T\left(\int_r^T\partial_x \sigma(u,S_u)  S_udu \right) dr \right)
\nonumber\\
&&= \lim_{T \to t}\frac{1}{ (T-t)^{2}} \mathbb{E}_t\left( \int_t^T\partial_x \sigma(u,S_u)  S_u\left(\int_t^udr\right) du \right).
\end{eqnarray}
Now, notice that
$$
S_u\partial_x\sigma(u,S_u)=\partial_X\hat{\sigma}(u,X_u),
$$
where $\hat{\sigma}$ denotes the local volatility function in terms of the log-price $X$. Then
\begin{eqnarray}
\label{onehalf}
&&\lim_{T \to t} \partial_kI_t(T,k^{*}) \nonumber\\
&& =\lim_{T \to t}\frac{1}{ (T-t)^{2}} \mathbb{E}_t\left( \int_t^T\partial_x \hat{\sigma}(u,X_u) \left(\int_t^u dr \right) du \right)\nonumber\\
&&=\frac{1}{2}\partial_X\hat{\sigma}(t,X_t).
\label{skewIVfinite}
\end{eqnarray}
\begin{remark}
The above is consistent with the one-half rule,  this heuristic relationship was introduced by Derman, Kani, and Zou (1996). According to this rule, in the short end,  the local volatility varies with the asset price twice as fast as the implied volatility varies with the strike. This rule has been proved for stochastic volatility models via different techniques (see Derman, Kani, and Zou (1996), Gatheral (2006), Lee (2001), or Al\`os and Garc\'ia-Lorite (2021)).
\end{remark}

\subsubsection{The rough local volatility case}
Now, assume that  Hypotheses \ref{Hyp1}, \ref{Hyp2}, and \ref{Hyp3}  hold for some $\gamma >0$.  Under these hypotheses the local volatility has a singularity at $(0,S_0)$, and then we refer to it as a 'rough local volatility'. Now, as
$$
D_{\theta}^{W}\sigma_u^2=2\sigma(u,S_u)\partial_x\sigma (u,S_u)D_{\theta}^{W}S_u
$$
and
\begin{eqnarray}
D_\theta^W D_r^W\sigma^2(u,S_u)
&=&2(\partial_x\sigma(u,S_u))^2D_\theta^WS_uD_r^WS_u\nonumber\\
&+&2\sigma(u,S_u)\partial_{SS}^2\sigma(u,S_u)D_\theta^WS_uD_r^WS_u\nonumber\\
&+&2\sigma(u,S_u)\partial_x\sigma(u,S_u)D_\theta^WD_r^WS_u,
\end{eqnarray}
a direct computation leads to $$(E(D_{\theta}^{W}\sigma_u^2)^p)^\frac{1}{p}\leq Cu^{-\gamma}$$ and $$(E(D_{\theta}^{W}D_r^W\sigma_u^2)^p)^\frac{1}{p}\leq Cu^{-2\gamma}.$$
Notice that, if $t>0$, $u^{-\gamma}<t^{-\gamma}$ and then Hypothesis \ref{Hyp6} holds with $H=\frac12$, which implies that the one-half rule is satisfied. Nevertheless, if  $t=0$, this hypothesis holds with $H=\frac12-\gamma$. Then, the one-half rules does not hold in this case, and we have the following result

\begin{theorem} 
\label{TeoremaLVIV}
Consider a local volatility model satisfying Hypotheses \ref{Hyp1}, \ref{Hyp2}, \ref{Hyp3}, \ref{Hyp4}, and \ref{Hyp7} with $t=0$. Then 
$$
\lim_{T \to 0} T^{\frac{1}{2}-H}\partial_kI_0(T,k^{*}) = \frac{1}{ \frac{3}{2}+H}\lim_{T \to 0}T^{\frac{1}{2}-H}\partial_x\hat\sigma(T,X_T),
$$
\label{skewlocal}
where $\hat\sigma$ denotes the local volatility function in terms of the log-price $X$ and $H=\frac12-\gamma$.
\end{theorem}
\begin{proof}
As
\begin{eqnarray}
D_r^W\sigma^2(u,S_u)&=&2\sigma(u,S_u)D_r^W(\sigma(u,S_u))\nonumber\\
&=&2\sigma(u,S_u)\partial_x\sigma(u,S_u)D_r^WS_u,
\end{eqnarray}
Then, a direct application of Theorem \ref{recallth} and Equation (\ref{derlv}) give us that 
\begin{eqnarray}
&&\lim_{T \to 0} T^{\frac{1}{2}-H}\partial_kI_0(T,k^{*}) \\
&&= \frac{1}{ 2\sigma(0,S_0)^2}\lim_{T \to 0}\frac{1}{ T^{\frac{3}{2}+H}}\int_0^T\left(\int_r^T 2\sigma(u,S_u)\partial_x\sigma(u,S_u)\sigma(r,S_r)S_rdu \right) dr. \end{eqnarray}
Now, because of the continuity of the local volatility function $\sigma$ and the asset price $S$ we can write
\begin{eqnarray}
&&\lim_{T \to 0} T^{\frac{1}{2}-H}\partial_kI_0(T,k^{*}) = \lim_{T \to 0}
\frac{1}{ T^{\frac{3}{2}+H}}\int_0^T\left(\int_r^T S_u\partial_x\sigma(u,S_u)du \right) dr.
\end{eqnarray}
Now, as
$$
S_u\partial_x\sigma(u,S_u)=\partial_X\hat{\sigma}(u,X_u),
$$
where $\hat{\sigma}$ denotes the local volatility function in terms of the log-price $X$. Then
\begin{eqnarray}
&&\lim_{T \to 0} T^{\frac{1}{2}-H}\partial_kI_0(T,k^{*}) \\
&& = \lim_{T \to 0}\frac{1}{ T^{\frac{3}{2}+H}}\int_0^T\left(\int_r^T \partial_x\hat\sigma(u,x_u) du \right) dr \\
&&= \lim_{T \to 0}\frac{1}{ T^{\frac{3}{2}+H}}\int_0^T u \partial_x\hat\sigma(u,X_u) du.
\end{eqnarray}
Now, a direct application of l'H\^opital rule gives us that
\begin{eqnarray}
&&\lim_{T \to 0} T^{\frac{1}{2}-H}\partial_kI_0(T,k^{*}) \\
&& = \lim_{T \to 0}\frac{1}{ T^{\frac{3}{2}+H}}\int_0^T\left(\int_r^T \partial_x\hat\sigma(u,X_u) du \right) dr \\
&&= \lim_{T \to 0}\frac{1}{ (\frac{3}{2}+H)}T^{\frac{1}{2}+H}T \partial_x\hat\sigma(T,X_T) du \\
&&=\frac{1}{ \frac{3}{2}+H}\lim_{T \to 0} T^{\frac{1}{2}-H}\partial_x\hat\sigma(T,X_T),
\end{eqnarray}
as we wanted to prove.
\end{proof}

\begin{remark}The above result recovers the recent results by Bourgey, De Marco, Friz, and Pigato (2022).
\end{remark}

\begin{remark} \label{skewpowerlaw}This relationship between the local and implied volatility implies that, given a model for the volatility process $\sigma$, the local and the implied skew slopes follow the same power law.
\end{remark}

\begin{remark} \label{12rough} The above result seems to be in contradiction with real market practice, where the one-half rule is observed (see Derman, Kani, and Zou (1996)). However, both facts are compatible. Even when a rough local volatility will satisfy the above $\frac{1}{H+\frac32}$ rule, local volatilities computed in real market practice are never 'rough', but they are regularized near maturity, as we can see in the following example.
\end{remark}

\begin{example}
In Figure \ref{example_8} we present the empirical ratio between the at-the-money skew slope of the implied and the local volatility corresponding to the EUROSTOXX50E index, with date July 14 2023. In the plot the shortest market maturity is $T=0.03$. Then, the computation of the local volatility from $T=0$ up to $T=0.03$ was computed via a regularization procedure \footnote{For example we can suppose a local volatility flat between 0 to first Monte Carlo, or even to use a discretization of the following asymptotic expression 
$$
\lim_{T \to 0} I_0(T,K)=\frac{\log(\frac{K}{S_T})}{\int_{S_T}^{k} \frac{1}{\sigma(T,u)} du}
$$}, which leads in practice to a regular local volatility function, with a (even high) finite skew slope limit. Then, the one-half rule is satisfied, as predicted in Section 3.2.1.

\newpage

\begin{figure}[h]
\centering
\includegraphics[width=8cm]{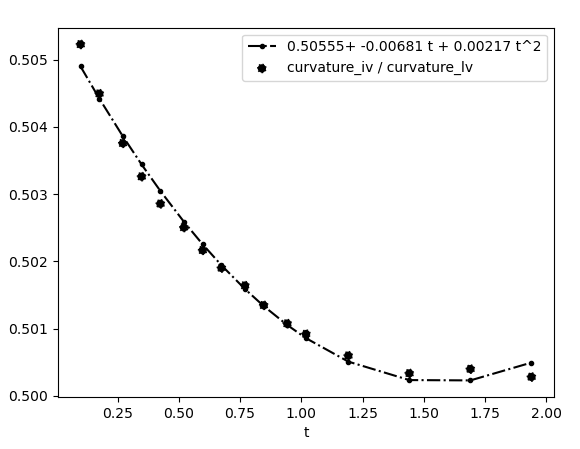}
\caption{Empirical $\frac{\partial_{k}^2I_0(T,k^{*})}{\partial^2_{x}\hat\sigma(T,X_T)}$ for market data of EUROSTOXX50E index.}  
\label{example_8}
\end{figure}
\end{example}

\section{The curvature}
Consider $t\in [0,T]$ and assume the following
\begin{h}\label{Hyp8}
There exists $H\in (0,1)$ such that, for all $t<\tau<\theta<r<u$,
$$
(E|(D_{\tau}^{W}D_{\theta}^{W}D_{r}^{W}\sigma_u^2 )|^p) ^\frac{1}{p}\leq C (u-r)^{H-\frac12}(u-\theta)^{H-\frac12}(u-\tau)^{H-\frac12}.
$$
\end{h}

\begin{h}\label{Hyp9}
There exists $H\in (0,1)$ such that the following quantities have a finite limit as $T\to t$
$$
\frac{1}{(T-t)^{2+2 H}} \mathbb{E}_t\int_t^T \left( \mathbb{E}_r\int_r^T D_r^{W}\sigma_u^2du \right)^2dr 
$$
$$
\frac{1}{(T-t)^{2+2 H}} \mathbb{E}_t\int_t^T \left( \int_s^T D_s^{W}\left( \sigma_r \int_s^T D_s^{W}\sigma_u^2du \right) dr \right) ds 
$$
$$
\frac{1}{ T^{2+2H}}\mathbb{E} \left( \int_0^T \int_s^T D_s^{W}\left( \sigma_r \int_r^T D_r^{W}\sigma_u^2du \right)drds \right).
$$
\end{h}

Let us recall the following result, that is an adaptation of Theorem 4.6 in Al\`os and Le\'on (2017) (see also Theorem 8.3.3 in Al\`os and Garc\'ia-Lorite (2021)).
\begin{theorem}\label{TeoremaLVVI} 
Take $t\in [0,T]$. Then, under Hypothesis \ref{Hyp5}-\ref{Hyp9}, and for every fixed $t \in [0,T],$ 
\begin{align}
\begin{split}
&\lim_{T \to t} (T-t)^{1-2H}\partial_{kk}^2I_t(T,k^{*}) \\
&=  \frac{1}{4\sigma_t^5}\lim_{T \to t} \frac{1}{(T-t)^{2+2 H}}\mathbb{E}\left(\int_t^T \left(\mathbb{E}_r\int_r^TD_r^{W}\sigma_u^2du\right)^2dr  \right)\\
& - \frac{3\rho^2}{ 2\sigma_t^5}\lim_{T \to t}\frac{1}{(T-t)^{3+2H}}\mathbb{E}\left( \int_t^T\left(\int_r^TD_r^{W}\sigma_u^2du \right) dr \right)^2  \\
&+\frac{\rho^2}{ \sigma_t^4}\lim_{T \to t}\frac{1}{ (T-t)^{2+2H}}\mathbb{E} \left( \int_t^T \int_s^T D_s^{W}\left( \sigma_r \int_r^T D_r^{W}\sigma_u^2du \right)drds \right).
\end{split}
\label{recall}
\end{align}
\label{thcurvature}
\end{theorem}
Similar arguments as in Section 3 give us that, for fixed $t>0$ and under hypotheses \ref{Hyp1}-\ref{Hyp4}, \ref{Hyp5}-\ref{Hyp9} with $\gamma=0$ the expression $\partial_{kk}^2I_t(T,k^{*}) $ tends to a finite quantity. Nevertheless, if  $t=0$, \ref{Hyp5}-\ref{Hyp9} hold for  $H=\frac12-\gamma$. In this case, we can prove the following result.

\begin{theorem} 
\label{curvature}
Under Hypothesis \ref{Hyp1}- \ref{Hyp4}, 
\begin{eqnarray}
&&\lim_{T \to 0} T^{1-2H}\partial_{kk}^2I_0(T,k^{*}) \nonumber\\
&&=\frac{1}{\sigma(0,S_0)}\lim_{u\to 0}u^{1-2H}( \partial_x\hat\sigma(u,X_u))^2\nonumber\\
&&\times\left[\frac{3}{(H+\frac32)(H+1)}-\frac{6}{(H+\frac32)^2}        +\frac{1}{2(H+1)}       \right]\nonumber\\
&+&\frac{1}{2(1+H)}\lim_{T \to 0}T^{1-2H}\partial^2_{xx}\hat\sigma(T,X_T),
\end{eqnarray}
where $H=\frac12-\gamma$.
\end{theorem}
\begin{proof}
Because of Theorem \ref{thcurvature}, we know that the limit
$$
\lim_{T \to 0} T^{1-2H}\partial_{kk}^2I_0(T,k^{*})
$$
is finite. Moreover, as local volatilities replicate vanilla prices, the result in Theorem  \ref{thcurvature} is also true if we replace the spot volatility $\sigma_u$ by the local volatility $\sigma (u,S_u)$. Then we can write
\begin{align}
\begin{split}
&\lim_{T \to 0} T^{1-2H}\partial_{kk}^2I_0(T,k^{*}) \\
&=  \frac{1}{4\sigma(0,S_0)^5}\lim_{T \to 0} \frac{1}{T^{2+2 H}}\mathbb{E}_t\left(\int_t^T \left(\mathbb{E}_r\int_r^TD_r^{W}\sigma^2(u,S_u)du\right)^2dr  \right)\\
& - \frac{3}{ 2\sigma(0,S_0)^5}\lim_{T \to 0}\frac{1}{ T^{3+2H}}\mathbb{E}_t\left( \int_t^T\left(\int_r^TD_r^{W}\sigma^2(u,S_u)du \right) dr \right)^2  \\
&+\frac{1}{ \sigma(0,S_0)^4}\lim_{T \to 0}\frac{1}{ T^{2+2H}}\mathbb{E}_t \left( \int_t^T \int_s^T D_s^{W}\left( \sigma (r,S_r) \int_r^T D_r^{W}\sigma^2(u,S_u)du \right)drds \right)\\
&=T_1+T_2+T_3.
\end{split}
\label{recall2}
\end{align}
Now the proof is decomposed into several steps.\\

{\it Step 1} Let us study the term $T_1$. Let us study the term $T_1$. As
\begin{eqnarray}
D_r^W\sigma^2(u,S_u)&=&2\sigma(u,S_u)\partial_x\sigma(u,S_u)D_r^WS_u,
\end{eqnarray}
and because of the continuity of $\sigma$,  and $S$, we get
\begin{eqnarray}
T_1&=&\frac{1}{\sigma(0,S_0)}\lim_{T \to 0} \frac{1}{T^{2+2 H}}\mathbb{E}\left(\int_0^T \left(\mathbb{E}_r\int_r^T\partial_x\sigma(u,S_u)S_udu\right)^2dr  \right)\\
&=&\frac{1}{\sigma(0,S_0)}\lim_{T \to 0} \frac{1}{T^{2+2 H}}\mathbb{E}\left(\int_0^T \left(\mathbb{E}_r\int_r^T\partial_x\hat\sigma(u,X_u)du\right)^2dr  \right).
\end{eqnarray}
Because of Theorem \ref{TeoremaLVIV} we know that
$$
u^{\frac12-H}\partial_x\hat\sigma(u,X_u)
$$
tends to a finite limit. Then we can write
\begin{eqnarray}
\label{termeT1a}
&&T_1\nonumber\\
&&=\frac{1}{\sigma(0,S_0)}\lim_{u\to 0} u^{1-2H}(\partial_x\hat\sigma(u,X_u)))^2\lim_{T \to 0} \frac{\int_0^T \left(\int_r^Tu^{H-\frac12}du\right)^2dr }{(T-t)^{2+2 H}}\nonumber\\
&&=\frac{1}{\sigma(0,S_0)}\lim_{u\to 0} u^{1-2H}(\partial_x\hat\sigma(u,X_u)))^2\lim_{T \to 0} \frac{\int_0^T[(T-t)^{H+\frac12}-(r-t)^{H+\frac12}]^2dr }{(H+\frac12)^2T^{2+2 H}}\nonumber\\
&&=\frac{1}{\sigma(0,S_0)(H+\frac12)^2}\left(1-\frac{2}{(H+\frac32)}+\frac{1}{2H+2}\right)\lim_{u\to 0} u^{1-2H}(\partial_x\hat\sigma(u,X_u)))^2.
\end{eqnarray}
Now, notice that
$$
1-\frac{2}{(H+\frac32)}+\frac{1}{2H+2}=\frac{(H+\frac12)^2}{(H+\frac32)(H+1)},
$$
and then
\begin{eqnarray}
\label{termeT1}
&&T_1=\frac{1}{\sigma(0,S_0)}\frac{1}{(H+\frac32)(H+1)}\lim_{u\to 0} u^{1-2H}(\partial_x\hat\sigma(u,X_u)))^2.
\end{eqnarray}
{\it Step 2}
In a similar way,
\begin{eqnarray}
\label{termeT2}
&&T_2\nonumber\\
&&= - \frac{6}{ \sigma (0,S_0)}\lim_{u\to 0} u^{1-2H}(\partial_x\hat\sigma(u,X_u)))^2\lim_{T \to 0}\frac{1}{ (T-t)^{3+2H}}\left(\int_0^T \left(\int_r^T u^{H-\frac12} du\right)dr \right)^2 \nonumber\\
&&= - \frac{6}{ \sigma (t,S_t)}\frac{1}{(H+\frac32)^2 }\lim_{u\to 0} u^{1-2H}(\partial_x\hat\sigma(u,X_u)))^2.
\end{eqnarray}
{\it Step 3} Let us now study the term $T_3$. Similar arguments as before allow us to write 
\begin{eqnarray}
\label{T3}
T_3&=&\frac{1}{ \sigma(0,S_0)^4}\lim_{T \to 0}\frac{1}{ T^{2+2H}}\mathbb{E} \left( \int_0^T \int_s^T D_s^{W}\left( \sigma (r,S_r) \int_r^T D_r^{W}\sigma^2(u,S_u)du \right)drds \right)\nonumber\\
&=&\frac{1}{ \sigma(0,S_0)^4}\lim_{T \to 0}\frac{1}{ T^{2+2H}}\mathbb{E} \left( \int_0^T \int_s^T \left(D_s^{W} \sigma (r,S_r) \int_r^T D_r^{W}\sigma^2(u,S_u)du \right)drds \right)\nonumber\\
&+&\frac{1}{ \sigma(0,S_0)^4}\lim_{T \to 0}\frac{1}{ T^{2+2H}}\mathbb{E} \left( \int_0^T \int_s^T \left( \sigma (r,S_r) \int_r^TD_s^{W} D_r^{W}\sigma^2(u,S_u)du \right)drds \right)\nonumber\\
&=&T_3^1+T_3^2.
\end{eqnarray}
As $D_s^{W} \sigma (r,S_r)=\partial_x\sigma (r,S_r)D_s^WS_r$, the continuity of $\sigma$ and $S$ allows us to write
\begin{eqnarray}
\label{T31}
T_3^1&=&\frac{2}{ \sigma(0,S_0)}\lim_{T \to 0}\frac{1}{ T^{2+2H}}\mathbb{E} \left( \int_0^T \int_s^T \left(\partial_x \sigma (r,S_r) S_r\int_r^T \partial_x\sigma(u,S_u)S_udu \right)drds \right)\nonumber\\
&=&\frac{2}{ \sigma(0,S_0)}\lim_{T \to 0}\frac{1}{ T^{2H}}\mathbb{E}\left( \int_0^T \int_s^T \left(\partial_x \hat\sigma (r,X_r) \int_r^T \partial_x\hat\sigma(u,X_u)du \right)drds \right).
\end{eqnarray}
Then, as $
u^{\frac12-H}\partial_x\hat\sigma(u,X_u)
$
has a finite limit, we get
\begin{eqnarray}
\label{Terme31}
T_3^1
&=&\frac{1}{ \sigma(0,S_0)}\lim_{u\to 0}u^{1-2H}( \partial_x\hat\sigma(u,X_u))^2\nonumber\\
&&\hspace{1cm}\times\frac{1}{ T^{2+2H}}\left( \int_0^T \left( \int_r^T u^{H-\frac12} du\right)^2dr\right)\nonumber\\
&=&\frac{1}{\sigma(0,S_0)}\frac{1}{(H+\frac32)(H+1)}\lim_{u\to 0}u^{1-2H}( \partial_x\hat\sigma(u,X_u))^2.
\end{eqnarray}
On the other hand, 
\begin{eqnarray}
\label{T312}
T_3^2&=&\frac{1}{ \sigma(0,S_0)^4}\lim_{T \to 0}\frac{1}{ (T-t)^{2+2H}}\mathbb{E} \left( \int_0^T \int_s^T \left( \sigma (r,S_r) \int_r^TD_s^{W} D_r^{W}\sigma^2(u,S_u)du \right)drds \right)\nonumber\\
&=&\frac{1}{ \sigma(0,S_0)^3}\lim_{T \to 0}\frac{1}{ T^{2+2H}}\mathbb{E} \left( \int_0^T \int_s^T \left(  \int_r^TD_s^{W} D_r^{W}\sigma^2(u,S_u)du \right)drds \right)\nonumber\\
\end{eqnarray}
 Now, notice that
\begin{eqnarray}
D_\theta^W D_r^W\sigma^2(u,S_u)
&=&2(\partial_x\sigma(u,S_u))^2D_\theta^WS_uD_r^WS_u\nonumber\\
&+&2\sigma(u,S_u)\partial_{SS}^2\sigma(u,S_u)D_\theta^WS_uD_r^WS_u\nonumber\\
&+&2\sigma(u,S_u)\partial_x\sigma(u,S_u)D_\theta^WD_r^WS_u
\end{eqnarray}
Then
\begin{eqnarray}
\label{T3121}
T_3^2&=&\frac{2}{\sigma (0,S_0)}\lim_{u\to 0}u^{1-2H}( \partial_x\hat\sigma(u,X_u))^2\nonumber\\
&&\times\frac{1}{ T^{2+2H}}\left( \int_0^T \int_s^T \ \int_r^T u^{2H-1}du drds \right)\nonumber\\
&+&2\lim_{T \to 0}\frac{1}{ (T-t)^{2+2H}}\mathbb{E}\left( \int_0^T \int_s^T \left(  \int_r^T[\partial^2_{SS}\sigma(u,S_u)S_u^2+\partial\sigma(u,S_u)S_u ]du \right)drds \right)\nonumber\\
&+&\frac{2}{\sigma (t,S_t)}\lim_{u\to 0}u^{1-2H}( \partial_x\hat\sigma(u,X_u))^2\nonumber\\
&&\times\frac{1}{ T^{2+2H}}\left( \int_0^T \int_s^T (r-t)^{H-\frac12}\left(  \int_r^Tu^{H-\frac12}du \right)drds \right)\nonumber\\
&=&\frac{1}{\sigma (0,S_0)}\lim_{u\to 0}u^{1-2H}( \partial_x\hat\sigma(u,X_u))^2\frac{1}{2(H+1)}\left(1+\frac{2}{(H+\frac32)}\right)\nonumber\\
&+&2\lim_{T \to 0}\frac{1}{ T^{2+2H}}\mathbb{E} \left( \int_0^T \int_s^T \left(  \int_r^T[\partial^2_{SS}\sigma(u,S_u)S_u^2+\partial_x\sigma(u,S_u)S_u ]du \right)drds \right)\nonumber\\
&=&\frac{1}{2\sigma(0,S_0)(H+1)}\left(1+\frac{2}{(H+\frac32)}\right)\lim_{u\to 0}u^{1-2H}( \partial_x\hat\sigma(u,X_u))^2\nonumber\\
&+&\lim_{T \to 0}\frac{1}{ T^{2+2H}}\mathbb{E} \left( \int_0^T\int_s^T\int_r^T \partial^2_{xx}\hat\sigma(u,X_u)u^2du drds \right).
\end{eqnarray}
Notice that, as all the other limits exist and are finite, the last term in the above equation is finite. Then, a direct application of l'H\^opital rule allows us to write
\begin{eqnarray}
\label{Terme32}
T_3^2&=&\frac{1}{2\sigma(0,S_0)(H+1)}\left(1+\frac{2}{(H+\frac32)}\right)\lim_{u\to 0}u^{1-2H}( \partial_x\hat\sigma(u,X_u))^2\nonumber\\
&+&\frac{1}{2(1+H)}\lim_{T \to 0}T^{1-2H}\partial^2_{xx}\hat\sigma(T,X_T).
\end{eqnarray}
Now, (\ref{termeT1}), (\ref{termeT2}), (\ref{Terme31}), and (\ref{Terme32}) give us that
\begin{eqnarray}
&&\lim_{T \to 0} T^{1-2H}\partial_{kk}^2I_0(T,k^{*}) \nonumber\\
&&=\frac{1}{\sigma(0,S_0)}\lim_{u\to 0}u^{1-2H}( \partial_x\hat\sigma(u,X_u))^2\nonumber\\
&&\times\left[\frac{3}{(H+\frac32)(H+1)}-\frac{6}{(H+\frac32)^2}        +\frac{1}{2(H+1)}       \right]\nonumber\\
&+&\frac{1}{2(1+H)}\lim_{T \to 0}T^{1-2H}\partial^2_{xx}\hat\sigma(T,X_T),
\end{eqnarray}
as we wanted to prove.
\end{proof}
\begin{remark}\label{curvature_regular_diffusion}
Notice that, if $H=\frac12$, the above reduces to 
\begin{eqnarray}
&&\lim_{T \to 0} T^{1-2H}\partial_{kk}^2I_0(T,k^{*}) \nonumber\\
&=&-\frac{1}{6\sigma(0,S_0)}\lim_{u\to 0}u^{1-2H}( \partial_x\hat\sigma(u,X_u))^2\nonumber\\
&+&\frac{1}{3}\lim_{T \to 0}T^{1-2H}\partial^2_{xx}\hat\sigma(T,X_T),
\end{eqnarray}
according to Equation (8.4.3) in Al\`os and Garc\'ia-Lorite (2021). On the other hand, in the uncorrelated case $\rho=0$ it reads
\begin{eqnarray}
&&\lim_{T \to 0}\partial_{kk}^2I_0(T,k^{*}) =\frac{1}{2(H+1)}\lim_{T \to 0}\partial^2_{xx}\hat\sigma(T,X_T).
\end{eqnarray}
In particular, if $\rho=0$ and $H=\frac12$, we get
\begin{eqnarray}
&&\lim_{T \to 0}\partial_{kk}^2I_0(T,k^{*}) =\frac{1}{3}\lim_{T \to 0}\partial^2_{xx}\hat\sigma(T,X_T),
\end{eqnarray}
according to the results in Hagan, Kumar, Lesniewski, and Woodward
(2002).
\end{remark}
\begin{remark}
As
$$
\lim_{T\to t}T^{1-2H}( \partial_x\hat\sigma(T,X_T))^2=4\lim_{T \to 0} T^{1-2H}(\partial_{k}I_0(T,k^{*}))^2
$$
the result in Theorem \ref{curvature} can be written as
\begin{eqnarray} \label{curvature_h}
&&\frac{1}{2(1+H)}\lim_{T \to 0}T^{1-2H}\partial^2_{xx}\hat\sigma(T,X_T)\nonumber\\
&&=\lim_{T \to 0} T^{1-2H}\partial_{kk}^2I_0(T,k^{*}) \nonumber\\
&&-\frac{4}{\sigma(0,S_0)}\lim_{T \to 0} T^{1-2H}(\partial_{k}I_0(T,k^{*}))^2\nonumber\\
&&\times\left[\frac{3}{(H+\frac32)(H+1)}-\frac{6}{(H+\frac32)^2}        +\frac{1}{2(H+1)} \right].
\end{eqnarray}
\end{remark}

\begin{remark} \label{ultima}This relationship between the local and implied volatility implies that, given a model for the volatility process $\sigma$, the local and the implied curvature satisfy the same power law.
\end{remark}

\section{Numerical Results}
This section is devoted to the numerical study of the relationship between the skews and curvatures of local and implied volatilities. Towards this end, we generate local volatility models satisfying the hypotheses required in the previous section, and we compare with the corresponding implied volatility behaviour.

\begin{example}[The skew]\label{skew_example}
Consider the rough Bergomi model with the volatility process given by
$$
\sigma_t=\sigma_0\exp\left(\nu\sqrt{2H}W_t^H-\frac12 \nu^2 t^{2H}\right),
$$
where $W_t^H=\int_0^t(t-s)^{H-\frac12}dW_s$ (see Bayer, Friz, and Gatheral(2016)). The set of parameters used in the simulation is the following $S_0=100$, $\nu =1.1$, $\sigma_0=0.3$, and  $\rho=-0.6$. We estimate the corresponding at-the-money implied and local volatility skews of a European call option as a function of maturity.  One can show that the following representation of $\sigma^{2}(T,K)$ holds
\begin{equation}\label{local_vol_representation}
\sigma^{2}(T,K) = \frac{\mathbb{E}_t\left(\sigma^{2}_t \phi(d(t,T,S_t,K)) \right)}{\mathbb{E}_t\left(\phi(d(t,T,S_t,K))  \right)}
\end{equation}
where $\phi(\cdot)$ is a density function of a standard Gaussian variable, $v_{t,T} =  \sqrt{\frac{\int_{t}^{T} \sigma^2_u du}{T-t}}$ and $d(t,T,S_t,K)=\frac{\log(\frac{K}{S_t})+ \frac{(T-t)v_{t,T}}{2} - \rho \int_{t}^{T} \sigma_u dB_u}{\sqrt{1-\rho^2} \sqrt{T-t}v_{t,T}}$. To compute the at-the-money skew of the local volatility, we compute $\partial_K$ in (\ref{local_vol_representation}) i.e
\begin{align}\label{skew_lv_representation}
&\partial_{K}\sigma(T,K) = \nonumber \\
& \frac{\mathbb{E}_t\left(\frac{\sigma^{2}_{T}d(t,T,x_t,K,v_{t,T})}{v^{2}_t K \sqrt{T-t}}\phi(d(t,T,x_t,K,v_{t,T})) \right) - \sigma^{2}(T,K)\mathbb{E}_t\left(\frac{d(t,T,x_t,K,v_{t,T}) \phi(d(t,T,x_t,K,v_{t,T}))}{K v^{2}_t \sqrt{T-t}} \right) }{2\sigma(T,K) \mathbb{E}_t\left(\frac{\phi(d(t,T,x_t,K,v_{t,T}))}{v_{t,T}} \right)}
\end{align}
To compute the at-the-money skew of the implied volatility we use finite difference method  \footnote{Another way to get the skew is to take the derivative with respect to $K$ in
\begin{equation*}
\mathbb{E}\left((S_{T}- K)^{+}\right) = BS(T,K,I(T,K)).
\end{equation*}
Then we get the following expression for the skew
\begin{equation*}
\partial_{k} I(T,K) = -\frac{\mathbb{E}\left(I(S_T > K) \right)- \partial_{K}BS(T,K,I(T,K)}{\partial_{\sigma}BS(T,K,I(T,K))}.
\end{equation*}
Notice that the term $\mathbb{E}\left(I(S_T > K) \right)$ can be estimated in the same simulation where we get the price of the option. We have checked both approaches and we have confirmed that they lead to identical results.} i.e
\begin{equation*}
\partial_K I(T,K) \approx \frac{I(T,K+h)-I(T,K-h)}{2h}.
\end{equation*}
Finally, we compute the ratio of the implied volatility skew over the local volatility skew for $H=\frac{1}{2}$ and $H=0.2$ and we show, that in the limit, 
$\frac{\partial_{k}I_0(T,k^{*})}{\partial_{x}\hat\sigma(T,X_T)}$ goes to $\frac{1}{H+\frac{3}{2}}=0.5$ and $\frac{1}{H+\frac{3}{2}}=0.588$, respectively. These coincide with the theoretical results provided in Theorem \ref{TeoremaLVIV} and Remark \ref{skewpowerlaw}. In Figure 2 we show the results that we have obtained from the Monte Carlo simulation
\begin{figure}[ht]
\begin{subfigure}{.5\textwidth}
  \centering
  \includegraphics[width=.9\linewidth]{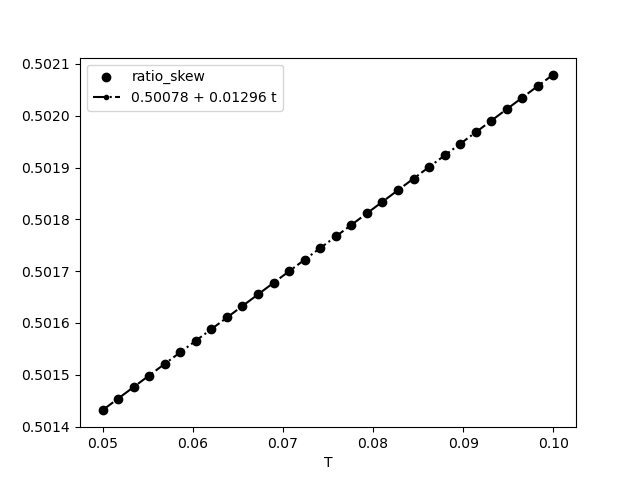}  
  \caption{$H=0.5$}
  \label{fig1:sub-first}
\end{subfigure}
\begin{subfigure}{.5\textwidth}
  \centering
  \includegraphics[width=.9\linewidth]{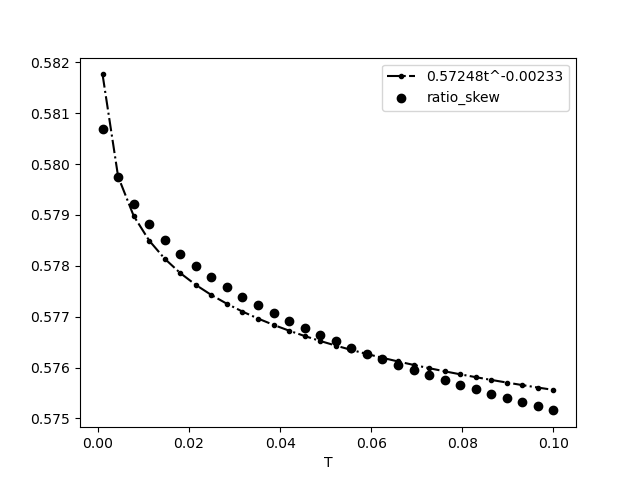}  
  \caption{$H=0.2$}
  \label{fig1:sub-second}
\end{subfigure}
\caption{$\frac{\partial_{k}I_0(T,k^{*})}{\partial_{x}\hat\sigma(T,X_T)}$ as a function of $T$}
\end{figure}

\end{example}

\begin{example}
The goal of this example is to show that the behaviour of the curvature in the SABR model
$$
dF_t=\sigma_t F_t^\beta (\rho dW_t +\sqrt{1-\rho^2} dB_t)
$$ 
with $\beta=1$, where the volatility is given by
$$
\sigma_t=\alpha\exp\left(-\frac{\nu^2}{2}t+\nu B_t\right),
$$
is in accordance with Remark \ref{curvature_regular_diffusion}. In order to show this we use the following approximation of the local volatility equivalent (see \cite{HKLW} for more details) for a log-normal SABR model
\begin{equation*}
\sigma(T,K) \approx \alpha \sqrt{1 + 2 \rho \nu y(K) + \nu^{2} y^{2}(K)}
\end{equation*}  
with $y(K)= \frac{\log(\frac{K}{S_0})}{\alpha}$. Then we have the following expressions for the skew and curvature
\begin{align*}
\partial_K \sigma(T,K)&= \frac{\alpha^{2}y^{\prime}(K)\left(\rho \nu + \nu^{2} y(K) \right)}{\sigma(T,K)} \\
\partial_{K,K} \sigma(T,K) &= \frac{\alpha^{2}y^{\prime \prime}(K)(\rho \nu + \nu y(K))+ (\alpha \nu y^{\prime}(K))^{2} -  (\partial_K \sigma(T,K))^{2}}{\sigma(T,K)}
\end{align*}
Therefore, we have that
\begin{equation}\label{curvature_log_lv}
\partial_{k,k}\hat{\sigma}(T,k) = K \partial_K \sigma(T,K) + K^{2}\partial_{K,K} \sigma(T,K).
\end{equation}
In addition, we use the approximation for the implied volatility proposed in \cite{HKLW} for the log-normal case
\begin{equation*}
I(T,K) \approx \alpha \frac{z}{x(z)}\left(1+ \left(\frac{1}{4}\rho \nu \alpha + \frac{2 - 3\rho^{2}}{24} \nu^{2} \right) T \right)
\end{equation*}
where $z=\frac{\nu}{\alpha} \log \left({\frac{S_0}{K}}\right)$ and $x(z)= \log \left(\frac{\sqrt{1-2\rho z + z^{2}}}{1-\rho}\right)$. It is straightforward to show that
\begin{align*}
\partial_{K} I(T,K) &= - \nu \frac{f^{\prime}(z)}{K}m(T) \\
\partial_{K,K} I(T,K) &= \left(\nu \frac{f^{\prime}(z)}{K^{2}} + \nu^{2} \frac{f^{\prime \prime}(z)}{\alpha K^{2}}\right)m(T)
\end{align*}
with $m(T)=1+ \left(\frac{1}{4}\rho \nu \alpha + \frac{2 - 3\rho^{2}}{24} \nu^{2} \right) T$ and $f(z)=\frac{z}{x(z)}$.
Therefore
\begin{equation}\label{curvature_log_iv}
\partial_{k,k} I(T,k^{*}) = K^{*} \partial_{K}I(T,K^{*}) + K^{2}\partial_{K,K}I(T,K^{*}).
\end{equation}
where $K^{*}=\exp(k^{*})$.
Firstly, note that by Remark \ref{curvature_regular_diffusion} the curvature of the local volatility is bigger than the curvature of the implied volatility. The difference between the curvatures for the SABR model in the short-end limit term is equal to $\frac{\rho^{2}\nu^{2}}{6 \alpha}$. In order to check it we run the Monte Carlo simulation with the following set of parameters for the SABR model $\alpha=0.3, \nu=0.6, \rho=-0.6$ and $S_0=100$ (notice than then  $\frac{\rho^{2}\nu^{2}}{6 \alpha}=\frac{0.6^4}{6 \times 0.3}=0.072$). The results are presented in Figure \ref{dos}. Recall that Remark \ref{curvature_regular_diffusion}, gives us that for $\rho=0$ 
$$
\lim_{T \to 0} \partial_{kk}^2I_0(T,k^{*}) =\frac{1}{3} \partial^2_{xx}\hat\sigma(T,X_T)
$$
This property can be observed in Figure \ref{3}.

\begin{figure}[ht]
\begin{subfigure}{.5\textwidth}
  \centering
  \includegraphics[width=1.0\linewidth]{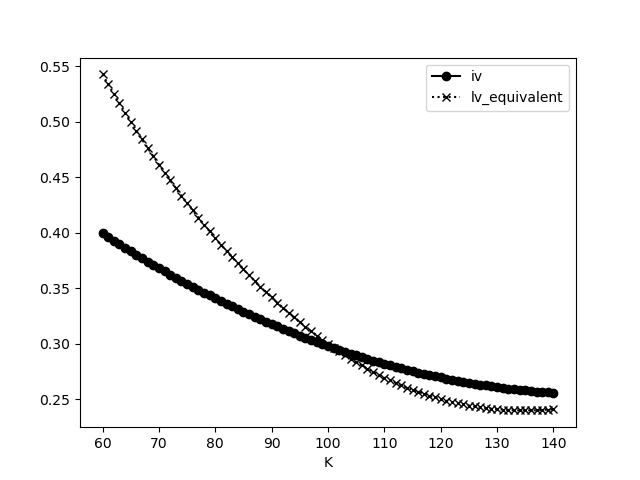}  
  \caption{Local volatility vs Implied volatility at $T=0.5$}
  \label{fig2:sub-first}
\end{subfigure}
\begin{subfigure}{.5\textwidth}
  \centering
  \includegraphics[width=1.0\linewidth]{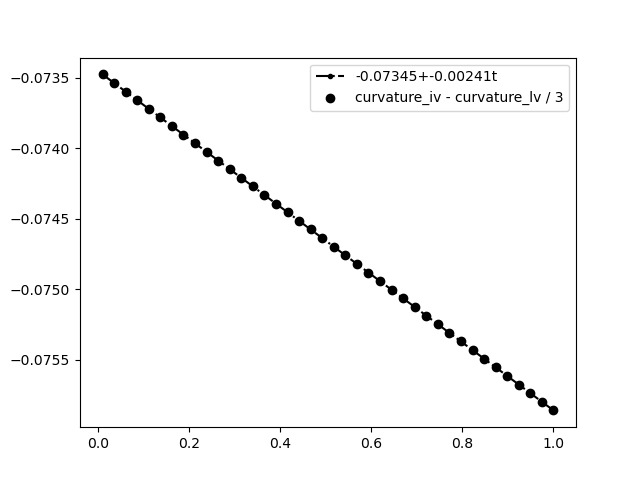}  
  \caption{$\partial_{kk}^2I_0(T,k^{*}) - \partial^2_{xx}\hat\sigma(T,X_T)$ as a function of $T$}
  \label{fig2:sub-second}
\end{subfigure}
\caption{Local volatility equivalent vs Implied volatility and differences between the curvatures at the short-term.}
\label{dos}
\end{figure}

\newpage

\begin{figure}[h]
\centering
\includegraphics[width=8cm]{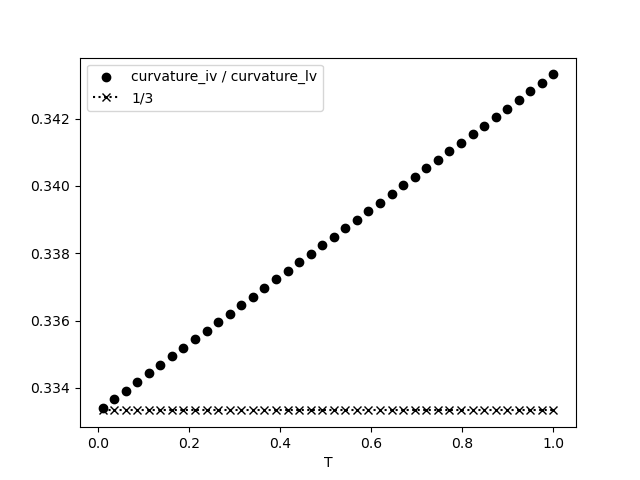}
\caption{$\frac{\partial_{kk}^2I_0(T,k^{*})}{\partial^2_{xx}\hat\sigma(T,X_T)}$ as a function of $T$.}  
\label{3}
\end{figure}
\end{example}

\begin{example}
In this example, we show that the power law followed by the curvature of the implied volatility and the equivalent local volatility is of order $(T-t)^{1-2H}$  when $0 \leq H < \frac{1}{2}$ for the rough Bergomi model. In order to compute the curvature of the implied volatility avoiding the noise of the Monte Carlo simulation, we use the approximation for the implied volatility proposed by the authors in \cite{MG}. The following dynamic for the underlying is assumed

\begin{align*}
\frac{S_t}{\beta(s_t)} &= \sigma_t dW_t \\
\frac{d\xi_t(s)}{\xi_t(s)} &= k(s-t) dB_t \quad t < s.
\end{align*}
where $\sigma_t=\sqrt{\xi_t(t)}$, $\beta$ is a positive continuous function, $\xi_t(s)=\mathbb{E}_t\left(\xi_s(s) \right)$ and $k(\tau)=\eta \sqrt{2H} t^{H-\frac{1}{2}}$. We must note that different choices of $H$ and $\beta(\cdot)$ lead to different models. For example, when $\beta(s)=s$ we recover the rough Bergomi model, which is our case of interest. The implied volatility approximation proposed in \cite{MG} is

\begin{equation}\label{G-M_Formula}
I(T,k)= I(T,0) \frac{|y(T,k)|}{\sqrt{G_{A}(y(k_{\beta}(k),T))}}
\end{equation}
where
\begin{align*}
y(T,k)&= \frac{k(T-t)}{U_t(T)}\\
U_t(T) &=\sqrt{\frac{1}{T-t}\int_{t}^{T} \xi_{t}(s)ds}\\
k &= \log(K) - k^{*} \\ 
k_{\beta}(k) &= \int_{s}^{k} \frac{ds}{\beta(s)}
\end{align*}
To estimate $\partial_{k} I(T,k)$ and $\partial_{k,k} I(T,k)$ we will use finite difference i.e
\begin{align*}
\partial_{k} I(T,k) &= \frac{I(T,k+h) - I(T,k-h)}{2h} \\
\partial_{k,k} I(T,k) &= \frac{I(T,k+h) - 2 I(T,k) + I(T,k-h)}{h^{2}}  
\end{align*}
On other hand, to estimate the local volatility equivalent and the skew we use (\ref{local_vol_representation}) and (\ref{skew_lv_representation}), respectively. The curvature is also estimated by finite differences i.e
\begin{equation*}
\partial_{K,K} \sigma(T,K) \approx \frac{\partial_K \sigma(T,K + h) - \partial_K \sigma(T,K - h)}{2h}.
\end{equation*}
In Figure 5 we can see that the power law for the curvature of the implied volatility and local volatility coincides. In the Monte Carlo simulation we assumed the following values of parameters $H=0.2$, $\nu=1.1$ and $\rho=-0.6$. The result justifies  Remark \ref{ultima}.

\begin{figure}[h]
\centering
\includegraphics[width=9cm]{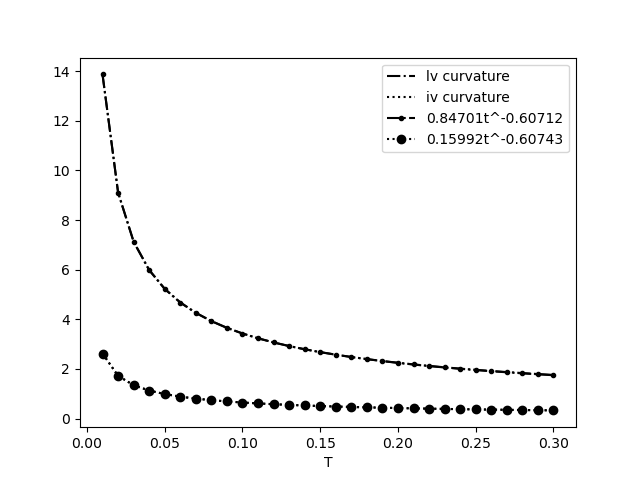}
\caption{Short-end curvatures of the local volatility and the at-the-money implied volatility.}  
\end{figure}

\end{example}

\end{document}